\documentclass[10pt,a4paper]{article}

\usepackage{comment}
\usepackage{todonotes}

\usepackage{amsthm}
\usepackage{amsmath}
\usepackage{amssymb}
\usepackage{mathrsfs}
\usepackage{amsbsy}

\theoremstyle{definition}
\newtheorem{dfn}{Definition}
\newtheorem{lem}{Lemma}
\newtheorem{exm}{Example}
\newtheorem{thm}{Theorem}
\newtheorem{cor}{Corollary}
\newtheorem{prop}{Proposition}

\newtheorem{claim}{Claim}

\newcommand{\s}[1]{\mathbf{#1}}
\newcommand{\RN}{\reals^\naturals}
\newcommand{\topo}[1]{\mathcal{#1}}
\newcommand{\ball}[2]{\mathcal{B}(#1,#2)}

\newcommand{\restricted}{\upharpoonright}
\newcommand{\norm}[2]{\lVert#1\rVert_{_{#2}}}
\newcommand{\spabscon}{\mathbf{1^1}}

\newcommand{\naturals}{\mathbb{N}}
\newcommand{\reals}{\mathbb{R}}
\newcommand{\rationals}{\mathbb{Q}}

\newcommand{\mach}[1]{\mathcal{#1}}  
\newcommand{\oramach}[3]{\mathcal{#1}^{^{#2}}(#3)}  

\newenvironment{myequation}{\begin{equation}\begin{aligned}}{\end{aligned}\end{equation}}

\title{Computability and Complexity over the Product Topology of Real Numbers}
\author{Walid Gomaa $^{1,2}$ \\
\small $^1$ Egypt Japan University of Science and Technology,\\ \small Alexandria, Egypt\\
\small $^2$ Faculty of Engineering,\\
\small Alexandria University, Alexandria, Egypt \\
\small walid.gomaa@ejust.edu.eg}

\begin{document}

\maketitle

\begin{abstract}

Kawamura and Cook have developed a framework for studying the computability and complexity theoretic problems over ``large'' topological spaces. This framework has been applied to study the complexity of the differential operator and the complexity of functionals over the space of continuous functions on the unit interval $C[0,1]$. 
In this paper we apply the ideas of Kawamura and Cook to the product space of the real numbers endowed with the product topology. We show that no computable norm can be defined over such topology. We investigate computability and complexity of total functions over the product space in two cases: (1) when the computing machine submits a uniformally bounded number of queries to the oracle and (2) when the number of queries submitted by the machine is not uniformally bounded. In the first case we show that the function over the product space can be reduced to a function over a finite-dimensional space.
However, in general there exists functions whose computing machines must submit a non-uniform number of queries to the oracle indicating that computing over the product topology can not in general be reduced to computing over finite-dimensional spaces.

\end{abstract}

keywoeds: computable analysis, second-order complexity theory, product topology, oracle Turing machine

\section{Introduction}

\label{sec:Introduction}

The theory of discrete/digital computation is very well-developed and understood. All powerful enough computational models have shown to be equivalent and a rigorous complexity and algorithmic theory have been developed. As a consequence there is \textit{the unified Church-Turing thesis} that is believed to hold.
However, when considering computation over continuous spaces, such as the real and complex numbers, the situation is absolutely the contrary.
Different computational models have been proposed ranging from direct extensions of Turing machines to algebraic models to analog computation. Many of these models are radically different; and there is neither a unified accepted computational model nor a unified Church-Turing thesis. Hence, continuous computation is still under-developed and not well-understood. Let alone the under-development of a rigorous complexity theory for such kind of computation. Only in recent years began the serious work and awareness of computability over general topological and metric spaces.

Computable analysis is a computational theory developed to address computation over general spaces. It had been developed since the early days of computer science and digital computation.
Computable analysis was introduced by A. Turing in 1936 \cite{TComputable}, A. Grzegorczyk in 1955 \cite{GComputable},
and D. Lacombe in 1955 \cite{LExtension} as an extension of classical discrete computability
by enhancing the normal Turing machine with oracles that provide access to the real-valued inputs.
It is a reductionist approach where the real number is deconstructed into
some finitary representation such as Cauchy sequences. Given a function $f\colon\reals\to\reals$, computability of $f$ in this context simply means the existence of a Turing machine
that when successively fed increasingly accurate representations of $x\in\reals$,
will be able to successively output increasingly accurate representation of the function value $f(x)$. Turing machines represent a discrete-time discrete-space model of continuous computation; they are finitary objects,
hence only a countable number of real functions are computable. Computable analysis is probably the most realistic approach to
continuous computation and hence
considered as the most suitable theoretical framework for numerical algorithms.
For a comprehensive treatment of the subject, especially from the
computability perspective, see \cite{WComputable}. See 
\cite{KComplexity} for a treatment of the complexity-theoretic investigations. For an extensive review of the algebraic characterizations of computable analysis, consult 
\cite{GAlgebraic}.

An approach to computable analysis is 
Type-Two Theory of Effectivity (TTE) which enables one
to extend computability theory from discrete spaces to
many continuous spaces arising in mathematical analysis
\cite{WComputable,BHWTutorial}.
On the other side, computational complexity
theory over continuous spaces is still in its infancy.
A theory applicable to the space of real numbers has
been developed by Ko and Friedman \cite{KFComputational,KComplexity}, and has
given many results. However, this theory is not readily
extendible to “larger” spaces such as the space $C[0, 1]$ of
continuous real functions defined over the unit interval,
and a more general, abstract theory is still lacking.
First approaches have been developed by Weihrauch \cite{WComputational} on metric spaces,
and by Schr\"oder \cite{SSpaces} who argues that in
order to express computational complexity in terms of
first-order time functions (as in the discrete setting), 
one must restrict to $\sigma$-compact spaces.
Recently Kawamura and Cook \cite{KCComplexity} developed a framework applicable to the space $C[0, 1]$ (which is not $\sigma$-compact), using higher order
complexity theory and in particular second-order
polynomials. In particular their theory enables them
to prove uniform versions of older results about the
complexity of solving differential equations, as well as
new results \cite{KLipschitz,KORZComputational,FGHQuery}.

In this article we begin a line of investigation of the computability and complexity theoretic-issues over the space (and subspaces) of sequences of real number $(x_i)_{_{i\in\naturals}}$, $x_i \in \reals$. We assume the product topology over $\RN$ as a computational representation and study the computational properties 
of such topology. We show that no computable norm can be defined over $\RN$ or subspaces of it with the product topology. We study total functions defined over $\RN$ and show that any such computable function submitting a uniformally bounded number of queries is equivalent to a function over $\reals^\ell$ for some $\ell \in \naturals$.
However, this is not the case when the machine computing 
the function queries its oracle in a non-uniform way (depending on either the input precision or the input itself from $\RN$). In such a case the second-order functional defined over $\RN$ can not be reduced to a first-order function over a finite-dimensional space.

Section \ref{sec:Introduction} is an introduction.
Section \ref{sec:Basic Definitions} gives the basic definitions of the product topology over the sequence space $\RN$ and the computability and second-order complexity notions in the sense of computable analysis.
Section \ref{sec:Computational Properties of the Product Topology} proves some basic computability results about $\RN$ and some of its subspaces endowed with the product topology. Section \ref{sec:Computing Total Functions over the product topology} investigates computability and complexity of of total functions over $\RN$ in two cases: (1) when the computing machine submits a uniformally bounded number of queries to the oracle and (2) when the number of queries submitted by the machine is not uniformally bounded.


\section{Basic Definitions}

\label{sec:Basic Definitions}

Let $\RN$ denote the space of sequences of real numbers, that is, the product space $\RN = \{\iota \colon \naturals \to \reals\}$. Let $\tau$ be the product topology over 
$\RN$. 
A basic open set in that topology is:
$U = I_0 \times \cdots \times I_k \times \RN$, 
where $k \in \naturals$ and $I_j \subseteq \reals$
is an open interval. Let $\topo X_p = (\RN,\tau)$ denote this topological space. As will be shown below this topology is not normable. 
So other topologies over subsets of $\RN$ 
are typically considered. 

Let $\mathbf{1}^\infty \subseteq \RN$ be the subset of bounded sequences and consider the topology on $\mathbf{1}^\infty$ induced by the uniform metric:
$\topo X_\infty = (\mathbf{1^\infty},\tau_{d_\infty})$, where $d_\infty(\s x, \s y) = \sup_{n \in \naturals} |x_n - y_n|$. This topology can also be generated from a basis $\mathcal{B}$, where each $B \in \mathcal{B}$ is of the form:
$B = \prod_{i\in \naturals} I_i$, where $I_i = (a,b)$ for some $a,b \in \reals$; so it is the \textit{box topology} on the set of bounded intervals. It is known that this latter space is not separable (hence, it does not form a computable metric space). Another subspace is the real Hilbert space $\mathbf{1}^2$ consisting of all sequences $\s x \in \RN$ 
such that $\sum_{i \in \naturals} x_i^2$ converges.
This subspace can be equipped with the metric $d_2(\s x,\s y) = \sum_{i \in \naturals} (x_i-y_i)^2$ to form the induced topology $\topo X_2 = (\mathbf{1}^2,\tau_{d_2})$.
The third subspace is concerned with  absolute convergence, it is the set of all sequences $\s x \in \RN$ such that $\sum_{i \in \naturals} |x_i|$ converges.
This subspace can be equipped with the metric
$d_1(\s x,\s y) = \sum_{i \in \naturals} |x_i-y_i|$ to form the induced topology $\topo X_1 = (\mathbf{1}^1, \tau_{d_1})$.

Let $\Sigma = \{0,1\}$ be a binary alphabet. The length of a finite string $u \in \Sigma^*$ is denoted by $|u|$.
An oracle Turing machine $\oramach{M}{\varphi}{u}$
takes an input $u \in \Sigma^*$; during its computation it is allowed to submit queries to an oracle $\varphi\colon\Sigma^*\to\Sigma^*$. A string function $\varphi\colon\Sigma^*\to\Sigma^*$ 
is said to be \textit{regular} if for all $u,v \in \Sigma^*$ the following holds: if $|u| \le |v|$, then $|\varphi(u)| \le |\varphi(v)|$. The size of a regular function $\varphi$ is the function $|\varphi| \colon \naturals \to \naturals$ with $|\varphi|(n) = |\varphi(0^n)|$. Second-order polynomials can be defined inductively as follows:

\begin{enumerate}
\item every $n \in \naturals$ is a second-order polynomial,
\item every first-order variable $x$ is a second-order polynomial,
\item if $P$ and $Q$ are second-order polynomials, then so are $P+Q$, $PQ$, and $f(P)$ where $f$ is a second-order variable.
\end{enumerate}

For example, the following is a second-order polynomial.

\begin{myequation}
\label{eqn:second-order polynomial}
P(f,x) := f(x+2) + f(f(x).f(x)) + x^2 + 4
\end{myequation}

A second-order polynomial $P$ can be viewed as an polynomial integer functional that essentially maps a function $f\colon\naturals\to\naturals$ to a function $P(f)\colon \naturals\to\naturals$.
For example, in Equation \eqref{eqn:second-order polynomial} let $f$ be defined as $f(x) = x^3$, then

\begin{myequation}
P(f,n)(x) = (x+2)^3 + (x^3.x^3)^3 + x^2 + 4
\end{myequation}

An oracle Turing machine $\mach{M}$ runs in polynomial time if there exists a second-order polynomial $P(X,Y)$
such that for any regular function $\varphi$ and any string $u \in \Sigma^*$, $\oramach{M}{\varphi}{u}$ 
halts in at most $P(|\varphi|,|u|)$ steps.
For more detailed information about representation by regular functions and second-order complexity theory the interested reader should consult \cite{KCComplexity}.
A sequence $\s x \in \RN$ can be represented by a regular function $\varphi_{\s x}\colon \Sigma^*\to\Sigma^*$, defined by:

\begin{myequation}
& \varphi_{\s x}(0^i,0^j) = 0^i\circ\alpha \\
& |[\alpha] - x_i| \le 2^{-j}
\end{myequation}

where $\alpha$ is a valid encoding of a rational number $[\alpha]$. (Notice that $\varphi$ can be assumed to take only one argument by the use of a proper pairing function.)

\begin{dfn}
A \textit{computable metric space} is a triple $(X,\mathcal{S},d)$ where:
\begin{enumerate}
\item $(X,d)$ is a separable metric space,
\item $\mathcal{S} = \{s_i \colon i \in \naturals\}$ is a countable dense subset of $X$, 
\item $d\colon X^2 \to \reals$ is computable.
\end{enumerate}
Similarly, we can define lower-semi computable and upper semi-computable metric spaces according to whether $d$ is lower-semi computable or upper semi-computable respectively.
\end{dfn}

\begin{prop}
The metric spaces $(\mathbf{1}^1,d_1)$ and $(\mathbf{1}^2,d_2)$ are lower semi-computable.
\end{prop}
\begin{proof}  
We first show that these spaces are separable.
Let $\mathcal{S} \subseteq \mathbf{1}^1$ be the set of sequences 
that are eventually zero and all of whose entries are in $\rationals$.
It is clear that $\mathcal{S}$ is countable, so remains to show that $\mathcal{S}$ is dense. Assume 
a point $\s x \in \mathbf{1}^1$ and $\epsilon > 0$.
We show that there exists some $\s s \in \mathcal{S}$ such that $d_1(s,x) < \epsilon$. Let $N \in \naturals$ be such that
$\sum _{i \ge N} |x_i| < \frac{\epsilon}{2}$.
For every $i=0,\ldots,N-1$, let $q_i \in \rationals$ be such that $|q_i - x_i| < \frac{\epsilon}{2N}$. 
Let $\s s \in \RN$ be such that 
$s_i = q_i$ for $i < N$, and zero otherwise; obviously, $\s s \in \mathcal{S}$.
Furthermore, $d_1(\s s,\s x) = \sum_{i \in \naturals} |s_i - x_i| = \sum_{i < N} |s_i - x_i| + \sum_{i \ge N} |s_i - x_i| = \sum_{i < N} |q_i - x_i| + \sum_{i \ge N} |x_i| < N \frac{\epsilon}{2N} + \frac{\epsilon}{2} = \epsilon$.
The proof for the space $\topo X_2$ is ver similar using the same countable dense set $\mathcal{S}$.
Assume $\s x, \s y \in \RN$. Fix $M \in \naturals$ and let $d_{\restricted M} = \sum_{i\le M} |x_i - y_i|$. It is clear that the function $d_{\restricted M}$ is computable, $d_{\restricted M} \le \sum_{i \in \naturals} |x_i-y_i| = d_1(\s x,\s y)$, and 
$d_{\restricted M} \to d_1(\s x,\s y)$ as $M \to \infty$. Hence, $d_1$ is lower semi-computable. Similarly for $d_2$.        
\end{proof}


\section{Computational Properties of the Product Topology}

\label{sec:Computational Properties of the Product Topology}

In this section we show some basic topological properties of the product topology $\topo X_p$ and its subspaces.
The following proposition shows that there is no
norm on $\RN$ that induces the product topology.

\begin{prop}
\label{prop: product topology not normable}
$\topo X_p$ is not normable.
\end{prop}
\begin{proof}
Assume that there exists a norm $\norm{\cdot}{}$ that
induces the product topology $\topo X_p$. Consider the unit ball $B = \{\s x \in \RN \colon \norm{x}{} < 1\}$. Then $\s 0 \in B$ and $B \in \tau$ ($B$ 
is open in the product topology).
There exists a constant $k \in \naturals$ such that $\{0\}^k \times \RN \subseteq B$.
Let $\s x$ be a sequence such that $x_{k+1} = 1$ and $0$ otherwise. Then $\s x \in B$ and so is every sequence $t\s x$ for any constant $t \in \reals$.
Let $t = 2/\norm{\s x}{}$, then $\norm{t \s x}{} = 
|t| \cdot  \norm{\s x}{} = 2$ which contradicts the fact that
$t\s x \in B$.
\end{proof}

Even though $\topo{X}_p$ is not normable it is metrizable.
Let $d$ be the usual metric on $\reals$: $d(x,y) = |x-y|$ for 
$x,y \in \reals$. Let $\bar{d}$ be the \textit{standard bounded metric} on $\reals$: $\bar{d}(x,y) = \min\{d(x,y),1\}$. Define the metric $D$ on $\RN$ as follows:

\begin{myequation}
\label{eqn:metric inducing the product topology on RN}
D(\s x, \s y) = \sup_{n \in \naturals}\left\{\frac{\bar{d}(x_n,y_n)}{n}\right\} ,\qquad \s x, \s y \in \RN
\end{myequation}

The following theorem from \cite{PFoundations} states that the metric $D$ induces the product topology over $\RN$.

\begin{thm}[Theorem 2.50 in \cite{PFoundations}]
\label{thm:product topology on RN is metrizable}
The metric $D$ defined in Eq. \eqref{eqn:metric inducing the product topology on RN} induces $\topo{X}_p$.
\end{thm}

The following proposition shows that any norm defined over $\RN$ induces a topology that is incomparable with the product topology.

\begin{prop}
\label{prop: norm induces topology over RN incomparable to the product topology}
Let $\norm{\cdot}{}$ be a norm defined over $\RN$ and let $\tau'$ be the topology induced by that norm. Then neither $\tau'$ is weaker nor stronger than $\tau$. 
\end{prop}
\begin{proof}
Assume $\tau'$ is weaker than $\tau$. Consider the unit ball $B = \{\s x \in \RN \colon \norm{x}{} < 1\}$.
Then $B \in \tau'$ and hence, by assumption, 
$B \in \tau$. By an argument similar to that in Proposition \ref{prop: product topology not normable}
we can show that this is impossible. This proves  
$\tau'$ is not weaker than $\tau$.
Now assume $\tau'$ is stronger than $\tau$ ($\tau'\supseteq \tau$).
\begin{claim}
\label{claim: no stronger than the product topology}
For all $i \in \naturals$, there exists a constant $c_i$ such that for any sequence $\s x$ the following holds:
if $\norm{\s x}{} < c_i$, then $|x_i| < 1$.
\end{claim}
\begin{proof}
Assume the claim does not hold. 
Then 

\begin{myequation}
\label{eqn: incomparable topology}
\exists i \in \naturals:\forall c > 0: \exists \s x \in \RN: \norm{\s x}{} < c \; and \; |x_i| \ge 1
\end{myequation}

Consider the set $A = \{\s y \in \RN \colon |y_i| < 1\}$. Then $A$ is open in the product topology, that is, $A \in \tau$, and hence, $A \in \tau'$.
Since 
$\s 0 \in A$, there exists $\epsilon > 0$ such 
that $\ball{\s 0}{\epsilon} \subseteq A$.
This means that $A$ contains all sequences  $\s z$ with $\norm{\s z}{} < \epsilon$, so we have $$ \forall \s z \in \RN: \norm{\s z}{} < \epsilon \; implies \; |z_i| < 1$$
However ,by replacing $c$ with $\epsilon$, this 
contradicts Eq. \eqref{eqn: incomparable topology}.
\end{proof}
\begin{claim}
\label{claim: relationship between norm and the individual components}
For all $\s x \in \RN$ and for all $i \in \naturals$ 
we have $\norm{\s x}{} > \frac{c_i}{2} |x_i|$. 
\end{claim}
\begin{proof}
Assume the contrary. Then there exists $i \in \naturals$ and $\s x \in \RN$ such that $\norm{x}{} \le \frac{c_i}{2} |x_i|$. 
Let $\beta$ be a positive number such that $\frac{c_i}{2 \norm{\s x}{}} \le \beta < \frac{c_i}{\norm{\s x}{}}$. Let $\s y = \beta \s x$, then $\norm{\s y}{} = \beta \norm{\s x}{} < \frac{c_i}{\norm{\s x}{}} \norm{\s x}{} = c_i$, 
so 
$\norm{\s y}{} < c_i$. We have by assumption

\begin{myequation}
& \norm{\s x}{} \le \frac{c_i}{2} |x_i| \\
& \beta \norm{\s x}{} \le \frac{c_i}{2} \beta |x_i| \\
& \norm{\beta \s x}{} \le \frac{c_i}{2} |\beta x_i| \\
& \norm{\s y}{} \le \frac{c_i}{2} |y_i| \\
& |y_i| \ge \frac{2}{c_i} \norm{\s y}{} = \frac{2}{c_i} \beta \norm{\s x}{} \ge \frac{2}{c_i} \frac{c_i}{2 \norm{\s x}{}} \norm{\s x}{} = 1
\end{myequation}

Hence, we have a sequence $\s y \in \RN$ such that 
$\norm{\s y}{} < c_i$, however, $|y_i| \ge 1$ which contradicts Claim \ref{claim: no stronger than the product topology}.

\end{proof}

Now define a sequence $\s z \in \RN$ as follows: for every $i \in \naturals$, let $z_i = \frac{2}{c_i} i$.
Then by Claim \ref{claim: relationship between norm and the individual components} we have:
$\norm{\s z}{} > \frac{c_i}{2} \frac{2}{c_i} i$, that is, $\norm{\s z}{} > i$ for every $i$ which is impossible. Hence,
$\tau'$ is not stronger than $\tau$.  

\end{proof}

The previous two propositions imply that no continuous, and hence computable, norm can be defined on $(\RN,\tau)$. 

\begin{cor}~
\label{cor:no continuous and computable norm over RN with the product topology}
\begin{enumerate}
\item No continuous norm can be defined over $(\RN,\tau)$, where $\tau$ is the product topology.
\item No computable norm can be defined over $(\RN,\tau)$.
\end{enumerate}
\end{cor}

Let $\s x \in \RN$ and let $\varphi_{\s x}$ be a representation of $\s x$. Let $\mach{M}$ be an oracle Turing machine computing over $\RN$.
Each query submitted by $\oramach{M}{\varphi_{\s x}}{n}$ is a pair of non-negative integers $\left<i,j\right>$ and $\varphi_{\s x}(\left<i,j\right>)$ returns a rational number $r$ such that 
$|r - x_i| \le 2^{-j}$.
The following proposition shows that no computable norm can be defined over the absolutely convergent
sequences $\spabscon$ with the subspace topology.

\begin{prop}
\label{prop:no computable norm over space of abs convergent sequences}
There is no computable norm over the subspace topology $(\spabscon,\tau_{_\spabscon})$.
\end{prop}
\begin{proof}
Assume that there exists such a norm $F$. Let $\mach{M}$
be an oracle Turing machine computing $F$. Consider 
the zero sequence $\s 0$. Then there exists an
oracle $\varphi_{\s 0}$ for $\s 0$ that would always answer zero to any query. 
Let $n \in \naturals$ be arbitrary and let
$N$ be the number of queries submitted by 
$\oramach{M}{\varphi_{\s 0}}{n}$. For $i=1,\ldots, N$, let
$<k_i,p_i>$ be the queries submitted by $\oramach{M}{\varphi_{\s 0}}{n}$.
Since $F$ is a norm, $F(\s 0) = 0$, and hence $\oramach{M}{\varphi_{\s 0}}{n}$ would output some 
positive number $\epsilon$ with $\epsilon \le 2^{-n}$. Let $k = \max_{i} k_i$ and define 
$A = \{\s x \in \spabscon \colon x_0=\cdots = x_k = 0\}$. For any $\s x \in A$, let $\varphi_{\s x}$ be a representation of $\s x$ such that $\varphi_{\s x}(<i,p>) = 0$ for any $i \le k$. So we have $\oramach{M}{\varphi_{\s x}}{n} = \oramach{M}{\varphi_{\s 0}}{n} = \epsilon$.
This implies that $F(\s x) \le 2^{-n+1}$ for all $\s x \in A$.
For every $\ell \in \naturals$ define the sequence $\s{y_\ell}$ such that $y_{\ell,k+1} = \ell$ and $0$ otherwise. Clearly, $\s{y_\ell} \in A$, so $F(\s{y_\ell}) \le 2^{-n+1}$.
We have $\s{y_\ell} = \ell \s{y_1}$, 
so $F(\s{y_\ell}) = \ell F(\s{y_1})$. $F(\s{y_1})$ is constant, so $F(\s{y_\ell})$ can be arbitrarily large with increasing $\ell$ which contradicts
the fact that $F(\s{y_\ell})$ is bounded by $2^{-n+1}$.
\end{proof}

The argument in the previous proposition can be generalized to the subspaces 
$\mathbf{1^\infty}$ and $\mathbf{1^2}$.

\begin{cor}~
\label{cor:computable norm over spaces 12 and 1infty}
\begin{enumerate}
\item There is no computable norm over the subspace topology $(\mathbf{1^\infty},\tau_{_\mathbf{1^\infty}})$.
\item There is no computable norm over the subspace topology $(\mathbf{1^2},\tau_{_\mathbf{1^2}})$.
\end{enumerate}
\end{cor}

\section{Computing Total Functions over $\topo{X}_p$}
 
\label{sec:Computing Total Functions over the product topology} 
 
Now we focus on the computability and complexity theoretic aspects of computing total functions $f \colon \RN \to \reals$. Given a function $f$ we always assume that the machine computing $f$ asks the minimum number of queries necessary to compute $f$ at the given input.

Assume a function $f$ and an oracle machine $\mach{M}$ computing $f$. Let $\mathcal{Q}_{\mach{M},\varphi_{\s x}}^n$ denote the set of all queries
submitted by $\oramach{M}{\varphi_{\s x}}{n}$ and let 
$\mathcal{Q}_{\mach{M},\varphi_{\s x}}^n(k)$ denote the $k$th
query 
where $k = 1,\ldots,|\mathcal{Q}_{\mach{M},\varphi_{\s x}}^n|$. Let $cord_{\mach{M},\varphi_{\s x}}^n$ be the projection of $\mathcal{Q}_{\mach{M},\varphi_{\s x}}^n$ on the first component, that is, $cord_{\mach{M},\varphi_{\s x}}^n$ is the set of the coordinates of the queries submitted by 
$\oramach{M}{\varphi_{\s x}}{n}$. Similarly, we let $cord_{\mach{M},\varphi_{\s x}}^n(k)$ denote the coordinate of the $k$th query. The following two lemmas show that the definition of $cord$ is independent of both the machine computing the function and the representation of the input sequence.

\begin{lem}
\label{lem:cordinates independent from representation of input}
Assume a function $f \colon \RN \to \reals$ that is computable by an oracle machine $\mach{M}$.
Let $\s x \in \RN$, then $cord_{\mach{M},\varphi_{\s x}}^n$ 
is the same for any representation $\varphi_{\s x}$ 
of $\s x$.
\end{lem}
\begin{proof}
If $f$ is constant, then the lemma holds trivially since $cord_{\mach{M},\varphi_{\s x}}^n = \emptyset$. So assume $f$ is not constant. Assume the lemma does not hold so there exists $\s x \in \RN$ and infinitely many $n \in \naturals$ such that $cord_{\mach{M},\varphi_{\s x}}^n \not= cord_{\mach{M},\varphi'_{\s x}}^n$ for some representations $\varphi_{\s x}$ and $\varphi'_{\s x}$ of $\s x$. 
Assume such $n \in \naturals$.
Let $A = cord_{\mach{M},\varphi_{\s x}}^n \cap cord_{\mach{M},\varphi'_{\s x}}^n$, $B = cord_{\mach{M},\varphi'_{\s x}}^n \setminus cord_{\mach{M},\varphi_{\s x}}^n$, and $C = cord_{\mach{M},\varphi_{\s x}}^n \setminus cord_{\mach{M},\varphi'_{\s x}}^n$. Then either one or both of $B$ and $C$ must be non-empty. Assume $B$ is non-empty. Since $B$ is non-empty and $\mach{M}$ submits the minimum number of queries necessary to compute $f(\s x)$ there exist
$\s y \in \RN$ such that: $y_i = x_i$ for all $i \in A \cup C$ and $f(\s x) \not= f(\s y)$.
Assume a number $m \in \naturals$ such that $|f(\s x) - f(\s y)| > 2^{-m}$. By continuity of $f$ we can choose $\s y$ such that $m$ is large enough with respect to $n$. Let $\varphi_{\s y}$ be a representation of $\s y$ such that $\varphi_{\s y}(0^i,.) = 
\varphi'_{\s x}(0^i,.)$ for all $i \in A \cup C$.
Construct a sequence $\s z \in \RN$ such that:
$z_i = x_i$ for all $i \in A \cup C$ and $z_i = y_i$ otherwise. Then $\s z$ can be represented by two regular functions: $\varphi_{\s z}$ with
$\varphi_{\s z}(0^i,.) = \varphi_{\s x}(0^i,.)$ for all $i \in A \cup C$ and $\varphi'_{\s z} = \varphi_{\s y}$.

Then we have:

\begin{myequation}
& |f(\s z) - \oramach{M}{\varphi_{\s z}}{n}| \le 2^{-n},\qquad \textit{by definition} \\
& |f(\s x) - \oramach{M}{\varphi_{\s x}}{n}| \le 2^{-n},\qquad \textit{by definition} \\
& \oramach{M}{\varphi_{\s x}}{n} = \oramach{M}{\varphi_{\s z}}{n}, \qquad \textit{by definition of } \varphi_{\s z}
\end{myequation}

Hence,

\begin{myequation}
|f(\s x) - f(\s z)| \le 2^{-n+1}
\end{myequation}

Similarly,

\begin{myequation}
& |f(\s z) - \oramach{M}{\varphi_{\s z}'}{n}| \le 2^{-n} ,\qquad \textit{by definition} \\
& |f(\s y) - \oramach{M}{\varphi_{\s y}}{n}| \le 2^{-n} ,\qquad \textit{by definition} \\
& \oramach{M}{\varphi_{\s y}}{n} = \oramach{M}{\varphi_{\s z}'}{n} , \qquad \textit{by definition of } \varphi'_{\s z}
\end{myequation}

Hence,
    
\begin{myequation}
|f(\s y) - f(\s z)| \le 2^{-n+1}
\end{myequation}

So

\begin{myequation}
|f(\s x) - f(\s y)| &\le |f(\s x) - f(\s z)| + |f(\s z) - f(\s y)| \\
& \le 2^{-n+1} + 2^{-n+1} \le 2^{-n+2} \\
\end{myequation}

By choosing $\s y$ such that $n \ge m+2$ we reach a contradiction.       
\end{proof}

As a consequence of this lemma we generalize our notation to $cord_{\mach{M},\s x}^n$.
The following lemma shows the independence of the coordinate set of the machine computing the relevant function.

\begin{lem}
\label{lem:cordinates independent from the machine computing the function}
Assume a function $f \colon \RN \to \reals$ that is computable by two oracle machines $\mach{M}_1$ and $\mach{M}_2$.
Then $cord_{\mach{M}_1,\s x}^n = cord_{\mach{M}_2,\s x}^n$.
\end{lem}
\begin{proof}
If $f$ is constant, then the lemma trivially holds. So assume $f$ is not constant and the lemma does not hold.
So there exist infinitely many $\s x \in \RN$ and 
infinitely many $n\in\naturals$ such that $cord_{\mach{M}_1,\s x}^n \not= cord_{\mach{M}_2,\s x}^n$. Choose such an $\s x \in \RN$ and a large enough $n \in \naturals$.
Let $A = cord_{\mach{M}_1,\s x}^n \cap cord_{\mach{M}_2,\s x}^n$, $B_1 = cord_{\mach{M}_1,\s x}^n \setminus cord_{\mach{M}_2,\s x}^n$, and
$B_2 = cord_{\mach{M}_2,\s x}^n \setminus                      cord_{\mach{M}_1,\s x}^n$. Either $B_i$ or both must be non-empty; assume $B_2$ is non-empty.
Since any machine computing $f$ submits the minimum number of queries required to generate a valid output, 
there exists $\s y \in \RN$ such that $y_i = x_i$ 
for all $i \not \in B_2$ and $|f(\s x) - f(\s y)| > 2^{-n+1}$. Let $d = \mathcal{M}_1^{^{\s y}}(n) = \mathcal{M}_1^{^{\s x}}(n)$ (the latter equality holds by definition of $\s y$ and the sets $A$ and $B_1$).
Then $|f(\s x) - f(\s y)| \le |f(\s x) -  \mathcal{M}_1^{^{\s x}}(n)| + |\mathcal{M}_1^{^{\s x}}(n) - \mathcal{M}_1^{^{\s y}}(n)| + |\mathcal{M}_1^{^{\s y}}(n) - f(\s y)| \le 2^{-n} + 0 + 2^{-n} = 2^{-n+1}$ which is a contradiction.
\end{proof}

As a consequence of this lemma we generalize the notation to $cord_{\s x}^n$.

\subsection{Bounded Number of Queries}
\label{Sec: Bounded Number of Queries}

In this section we give a characterization of the functions $f\colon \RN \to \reals$ which can be computed by Turing machines that submit only a uniformally bounded number of queries to the oracle regardless of the input sequence $\s x$ and the input precision $n$. 
We start by showing uniform boundedness with respect to a given input sequence; that is, given $\s x \in \RN$ there exists $M \in \naturals$ such that 
$cord_{\s x}^n(i) \le M$ for all $n$ and all $i$.

\begin{lem}
\label{lem:coordinates are bounded}
Assume a function $f\colon \RN \to \reals$ that is computable by an oracle Turing machine $\mach{M}$ that submits at most $\ell$ queries for some $\ell \in \naturals$. Then for each $\s x \in \RN$ we have $cord_{\s x}^n$ is uniformally bounded in $n$.
\end{lem}
\begin{proof}
If $f$ is constant, then the lemma trivially holds. So assume $f$ is not constant. Fix $\s x \in \RN$.
Let $B = \{i  \le \ell\colon \limsup_{n\to\infty}cord_{\s x}^n(i) < \infty\}$ and let $U = \{i \le \ell\colon \limsup_{n\to\infty}cord_{\s x}^n(i) = \infty\}$. Let $L \in \naturals$ such that for any 
$i \in B$ and $n \in \naturals$ we have $cord_{\s x}^n(i) \le L$. If $U = \emptyset$, then we are done. So assume $U$ is non-empty. Since the machine infinitely often submits a query beyond $L$, we can assume -without loss of generality- the existence of some $\s y \in \RN$ such that $x_i = y_i$ for all $i \le L$ and 
$f(\s x) \not= f(\s y)$. For every $n \in \naturals$
let $j_n = \max\{L,\min\{cord_{\s x}^n(i)\colon i \in U\}-1\}$. Define the following sequence in $\RN$: $(\s y^n)$, where $y^n_i = y_i$ for all $i \le j_n$ and $y^n_i = x_i$ otherwise. Then $(\s y^n) \to \s y$. For sufficiently large $n$ we have by definition of $\s y^n$: $\oramach{M}{\s y^n}{n} = \oramach{M}{\s x}{n}$, that is the machine $\mach{M}$ can not distinguish between $\s y^n$ and $\s x$ at precision $2^{-n}$.
Then $|f(\s x) - f(\s y^n)| \le 2^{-n+1}$, so

\begin{myequation}
& \lim_{n\to\infty} |f(\s x) - f(\s y^n)| \le \lim_{n \to \infty} 2^{-n+1} \\
& |f(\s x) - f(\s y)| \le \lim_{n \to \infty} 2^{-n+1} = 0 
\end{myequation}

So  $f(\s x) = f(\s y)|$ which is a contradiction 
\end{proof}

The next lemma is an extension of the previous one to show that $cord_{\s x}^n$ is fixed.

\begin{lem}
\label{lem:coordinates are bounded and fixed}
Assume the premises of Lemma \ref{lem:coordinates are bounded}. Then for each $\s x \in \RN$ we have $cord_{\s x}^n$ is fixed irrespective of $n$.
\end{lem}
\begin{proof}
Assume $f$ is not constant and let $\mach{M}$ be a Turing machine computing $f$. Assume the lemma does not hold.  Let $\s x \in \RN$ be a counterexample to the lemma and let $\alpha \in \naturals$ be a uniform bound on 
$cord_{\s x}^n$; its existence is determined by Lemma \ref{lem:coordinates are bounded}. Since $\alpha$ is finite there exist two infinite disjoint sets $N_1,N_2 \subseteq \naturals$
such that $cord_{\s x}^m = cord_{\s x}^n$ for all
$m,n \in N_1$, $cord_{\s x}^m = cord_{\s x}^n$ for all $m,n \in N_2$, and $cord_{\s x}^m \not= cord_{\s x}^n$ for any $m \in N_1$ and $n \in N_2$.
For $m \in N_1$ and $n \in N_2$ let $A = cord_{\s x}^m \cap cord_{\s x}^n$, $B = cord_{\s x}^m \setminus cord_{\s x}^n$, and $C = cord_{\s x}^n \setminus cord_{\s x}^m$. Then at least one of $B$ and $C$ must be non-empty. Assume $B$ is non-empty. There exists $\s y \in \RN$ such that $f(\s x) \not= f(\s y)$ and $y_j \not= x_j$ for some $j \in B$ and $y_i = x_i$ for all $i \in \naturals \setminus \{j\}$.
Let $k \in \naturals$ be minimal such that $|f(\s x) - f(\s y)| > 2^{-k}$. Let $k' \in N_2$ be such that 
$k' > k+2$. Then we must have: $\oramach{M}{\s x}{k'} = \oramach{M}{\s y}{k'}$ which implies that 
$|f(\s x) - f(\s y)| \le |f(\s x) - \oramach{M}{\s x}{k'}| + |\oramach{M}{\s x}{k'} - \oramach{M}{\s y}{k'}| + |\oramach{M}{\s y}{k'} - f(\s y)| \le 2^{-k'} + 0 + 2^{-k'} = 2^{-k'+1} < 2^{-k-1} < 2^{-k}$
which is a contradiction.
\end{proof}
 
As a consequence of Lemma \ref{lem:coordinates are bounded and fixed} we can simplify the notation  $cord_{\s x}^n$ to $cord_{\s x}$ (only for the case of bounded number of queries). The following proposition uniformally generalizes the previous two lemmas 
to the whole domain of $f$, that is, the whole space $\RN$.

\begin{prop}
\label{prop:coordinates are the same for all input}
Assume a function $f\colon \RN \to \reals$ that is computable by an oracle Turing machine that submits at most $\ell$ queries for some $\ell \in \naturals$. Then $cord_{\s x}$ is the same for all $\s x \in \RN$, that is, $cord_{\s x}(i)$ is the same for all $\s x \in \RN$ and all $i=1,\ldots,\ell$.
\end{prop}
\begin{proof}
If $f$ is constant, then the proposition trivially holds. So assume $f$ is not constant. Assume the proposition does not hold. To simplify the discussion we assume that, without loss of generality, $\ell = 2$. It is clear that $cord_{\s x}(1)$ is the same for all $\s x \in \RN$ since it is the first query submitted by the machine computing $f$. Let $i = cord_{\s x}(1)$. By assumption that the proposition does not hold there exists a non-constant computable function $\iota\colon\reals\to\naturals \cup \{-1\}$ such that $\iota(x_i) = cord_{\s x}(2)$ (if $\iota(x_i) = -1$, this means the machine computing $f$ submits only one query). (This function $\iota$ does the following: it takes a real number $a$, simulates $\mach{M}^{^{\s x}}$ for some sequence $\s x$ with $x_i=a$, when $\mach{M}^{^{\s x}}$ submits its second query with coordinate $k$, $\iota$ just returns $k$ and stops.)
Then $\iota$ must be continuous with respect to the Euclidean topology on $\reals$ and the discrete topology on $\naturals \cup \{-1\}$. Choose some $k \in \naturals$ such that $\iota^{-1}(k) \not=\emptyset$, then $\iota^{-1}(k) \supseteq (a,b)$. 
Assume, without loss of generality, that $b$ is a finite real and it is maximal such that $(a,b) \subseteq \iota^{-1}(k)$. Since the topology on $\naturals$ is the discrete topology, $\iota$ is an open map. There exists $\epsilon > 0$ 
such that $\iota(u) = k$ for any $u \in (b-\epsilon,b)$ 
and $\iota(v) = k' \not= k$ for any $v \in (b,b+\epsilon)$. Assume also, without loss of generality that $k' \not= -1$. Let $\s x,\s y \in \RN$ be such that $f(\s x) \not= f(\s y)$, $x_i = y_i = b$, and $y_k = x_k$ (note that $b\not=-1$ and $cord_{\s x}(2) \not= k$ and $cord_{\s y}(2) \not= k$).
Let $p$ be such that $|f(\s x) - f(\s y)| > 2^{-p}$.
Now consider the computation of $\oramach{M}{\s x}{p+1}$ and $\oramach{M}{\s y}{p+1}$ and assume that the oracle always approximate the $i^{th}$ component of the sequence $b$ from the left (that is it always answers with rational numbers $r < b$). This means that the coordinate of the second query will always be $k$, and hence, the oracle machine will not be able to distinguish between $\s x$ and $\s y$. 
So  $\oramach{M}{\s x}{p+1} = \oramach{M}{\s y}{p+1}$ which implies that $|f(\s x) - f(\s y)| \le 2^{-p}$ which is a contradiction.
\end{proof}

This last proposition motivates us to have the more generic notation $cord_f$ to denote the coordinate set for the whole domain of $f$.
Finally, the following theorem gives a characterization
of functions over $\RN$ whose computing machines submit only a uniformally bounded number of queries.

\begin{thm}
\label{thm:Characterizing functions with bounded number of queries}
Assume a function $f\colon \RN \to \reals$ that is computable by an oracle Turing machine that submits at most $\ell$ queries for some $\ell \in \naturals$. Then there exists a computable function $\phi \colon \reals^\ell \to \reals$ and a sequence of non-negative integers $i_1,\ldots,i_\ell$ such that $f(\s x) = \phi(x_{i_1},\ldots,x_{i_\ell})$.
\end{thm}
\begin{proof}
This is a direct consequence of Proposition \ref{prop:coordinates are the same for all input}. Let $\mach{M}$ be an oracle Turing machine computing $f$. $\varphi$ can be computed as follows. Assume an oracle Turing machine $\oramach{N}{x_1,\ldots,x_\ell}{n}$ that does the following:
\begin{itemize}
\item $\oramach{N}{x_1,\ldots,x_\ell}{n}$ simulates $\oramach{M}{\s x}{n}$ for some arbitrary sequence $\s x$,
\item when $\mach{M}$ submits its $i^{th}$ query, for $i=1,\ldots,\ell$, $\mach{N}$ submits a query to its own oracle $x_i$, and returns the response to $\mach{M}$,
\item when $\mach{M}$ halts with $r$ on the output tape, $\mach{N}$ writes $r$ on its output tape and halts.
\end{itemize}
\end{proof}

The following corollary is the complexity-theoretic version of the previous theorem.

\begin{cor}
\label{cor: complexity-theoretic - bounded number of queries}
Assume a function $f\colon \RN \to \reals$ that is computable by a polynomial-time oracle Turing machine that submits at most $\ell$ queries for some $\ell \in \naturals$. Then there exists a polynomial-time computable function $\phi \colon \reals^\ell \to \reals$ and a sequence of non-negative integers $i_1,\ldots,i_\ell$ such that $f(\s x) = \phi(x_{i_1},\ldots,x_{i_\ell})$.
\end{cor}

\subsection{Unbounded Number of Queries}

In this section we show that generally the number of queries accessed by a Turing machine computing a function $f\colon\RN\to\reals$ can be arbitrarily large with respect to the function $f$ or even with respect to single inputs of that function.

\begin{lem}
\label{lem:oracle queries not bounded for single sequence}
There exists a computable function $f\colon\RN\to\reals$ such that $cord_{\s x}^n$ is not bounded for any $\s x \in \RN$.
\end{lem}
\begin{proof}        
Consider the function $\alpha\colon\reals\to\reals$
defined by: $\alpha(x) = \frac{|x|}{1+|x|}$. Then the range of $\alpha$ is $[0,1)$. Now define the function 
$f\colon \RN\to\reals$ as follows:

\begin{myequation}
f(\s x) = \sum_{k=0}^\infty \frac{\alpha(x_k)}{2^k}
\end{myequation}

To compute $f$ up to a precision $2^{-n}$, the oracle machine needs to access the first $(n+2)$ elements of the sequence $\s x$. Hence, $cord_{\s x}^n$ is not bounded with respect to $n$.
\end{proof}

In the previous lemma we see that $cord_f^n$
is the same for fixed $n$ (irrespective of the input sequence $\s x$). The following proposition gives an example of a function $f$ with non-fixed $cord_f^n$.

\begin{prop}
\label{prop:}
There exists a computable function $f\colon\RN\to\reals$ such that:
\begin{enumerate}
\item Given $\s x \in \RN$, $cord_{\s x}^n$ is not fixed with respect to $n$
\item Given $n \in \naturals$, $cord_f^n$ is not well defined
\end{enumerate}
\end{prop}
\begin{proof}
Define the function

\begin{myequation}
f(\s x) = \sum_{k=0}^\infty \frac{\alpha(x_k)}{2^{k + |x_0|}}
\end{myequation}

where $\alpha$ is the function defined in Lemma \ref{lem:oracle queries not bounded for single sequence}.
Given a fixed $\s x$, it is clear from Lemma 
\ref{lem:oracle queries not bounded for single sequence} that $cord_{\s x}^n$ is not bounded with respect to $n$. Now fix some $n \in \naturals$, then it is clear that the length of the initial segment of the input sequence queried by the machine depends on input sequence itself, in particular, on the first element of the sequence. Hence, $cord_f^n$ is not fixed and therefore not well defined.

\end{proof}

\bibliographystyle{abbrv}
\bibliography{references}

\end{document}